\documentclass[11pt]{article}
\usepackage[utf8]{inputenc}

\usepackage{graphicx}
\usepackage{authblk}
\usepackage{xcolor}
\usepackage{url}
\usepackage{amsmath,amsthm, amssymb,amsfonts}
\usepackage{thmtools}
\usepackage{thm-restate}
\usepackage{fullpage}
\usepackage{tikz}
\usetikzlibrary{intersections}

\usepackage{verbatim}

\def \NPP/{$\textsc{Nested Polytope Problem}$}
\def \NPPshort/{$\textsc{NPP}$}
\def \EXACTNMF/{\ensuremath{\textsc{EXACT-NMF}}}
\def \NMF/{\ensuremath{\textsc{Non-negative Matrix Factorization}}}
\def \NMFshort/{\ensuremath{\textsc{NMF}}}
\def \RNMF/{\ensuremath{\textsc{RESTRICTED-NMF}}}
\def \IntSimpl/{\ensuremath{\textsc{Intermediate Simplex Problem}}}
\def \ER/{\ensuremath{\exists \mathbb{R}}}
\def \ETR/{\ensuremath{\textsc{ETR}}}
\def \AGP/{\ensuremath{\textsc{Art Gallery Problem}}}
\def \ETRINVA/{\ensuremath{\textsc{ETR-INV-array}}}
\def \ETRINVS/{\ensuremath{\textsc{ETR-INV-system}}}
\def \sol/{\ensuremath{\mbox{sol}}}
\def \NP/{\ensuremath{\mbox{NP}}}
\def \almostequal/{firmly-equal}

\newcommand {\R}{\ensuremath{\mathbb{R}}}
\newcommand {\Q}{\ensuremath{\mathbb{Q}}}
\newcommand {\Z}{\ensuremath{\mathbb{Z}}}
\newcommand {\N}{\ensuremath{\mathbb{N}}}
\renewcommand {\S}{\ensuremath{\mathcal{S}}}
\newcommand {\A}{\ensuremath{\mathcal{A}}}
\newcommand {\rank}{\ensuremath{\textrm{rank}}}
\newcommand {\conv}{\ensuremath{\textrm{conv}}}

\usepackage[super]{nth}

\newtheorem{theorem}{Theorem}
\newtheorem*{theorem*}{Theorem}
\newtheorem{lemma}[theorem]{Lemma}

\newtheorem{corollary}[theorem]{Corollary}
\newtheorem{prpty}[theorem]{Property}

\theoremstyle{definition}

\newtheorem{remark}[theorem]{Remark}
\newtheorem*{remark*}{Remark}

\newtheorem{example}[theorem]{Example}

\newtheorem{obs}[theorem]{Observation}

\title{A Universality Theorem for Nested Polytopes}
\author[1]{Michael G. Dobbins}
\author[2]{Andreas Holmsen}
\author[3]{Tillmann Miltzow}
\affil[1]{Binghamton University, USA}
\affil[2]{KAIST, Korea}
\affil[3]{Utrecht University, Netherlands}
\date{July 2018}

\begin{document}

\maketitle

\begin{abstract}
    In a nutshell, we show that polynomials and 
    nested polytopes are
    topological, algebraic and algorithmically equivalent.

    Given two polytops $A\subseteq B$ and a number $k$, 
    the \NPP/ (\NPPshort/) asks, if there exists a 
    polytope $X$ on $k$ vertices such that
    $A\subseteq X \subseteq B$.
    The polytope $A$ is given 
    by a set of vertices and
    the polytope $B$ is given 
    by the defining hyperplanes.
    We show a universality theorem for \NPP/.
    Given an instance $I$ of the \NPPshort/, 
    we define the solutions set
    of $I$ as
    \[V'(I) = \{(x_1,\ldots,x_k)\in \R^{k\cdot n} :
 A\subseteq \conv(x_1,\ldots,x_k) \subseteq B\}.\]
    As there are many symmetries, induced by
    permutations of the vertices, we will consider
    the \emph{normalized} solution space $V(I)$.
    
    Let $F$ be a finite set of polynomials, 
    with bounded solution space.
    Then there is an instance $I$ of the \NPPshort/, which has 
    a rationally-equivalent normalized solution space $V(I)$.
    
  Two sets $V$ and $W$ are \emph{rationally equivalent} 
  if there exists a homeomorphism
  $f : V \rightarrow W$ 
  such that both $f$ and $f^{-1}$ are given by rational functions.
  A function $f:V\rightarrow W$ is a homeomorphism, if it is
  continuous, invertible and its inverse is continuous as well.
    
    As a corollary, we show that \NPPshort/ is 
    \ER/-complete.
    This implies that unless $ \ER/ = \NP/$, the 
    \NPPshort/ is not contained in the complexity class NP.
    Note that those results already follow from a recent paper
    by Shitov~\cite{shitov2016universality}. 
    Our proof is geometric and arguably easier.
\end{abstract}

\section{Introduction}

\paragraph*{Definition.}
In the \NPP/ (\NPPshort/), we are given two polytopes 
$A\subseteq B \subset \R^{n}$
and a number $k\in \N$
and we ask, whether there exists a polytope
$A\subseteq X \subseteq B$ with $k$ vertices.
To be more precise the \emph{inner polytope} 
$A$ is specified by its vertices
and the \emph{outer polytope} 
$B$ is specified by its facets.
Given an instance $I = (A,B,k)$, we denote
by 
 \[V'(I) = \{(x_1,\ldots,x_k)\in \R^{k\cdot n} :
 A\subseteq \conv(x_1,\ldots,x_k) \subseteq B\}.\]
the set of \emph{solutions}.
Here $\conv(x_1,\ldots,x_k)$ denotes the convex hull
of the points $x_1,\ldots,x_k$.
Given a permutation $\pi : [k]\rightarrow [k]$,
we can for every solution $x\in V'(I)$ get a new 
solution denoted by $x_\pi$.
Note that $V'(I)$ has a lot of symmetries as every 
permutation of the vertices of a valid solution
yields again a valid solution.
We say two solutions $x,y$ are \emph{permutation-equivalent}
if there exists a permutation $\pi$ of
the vertices, so that $x_\pi = y$.
We denote this by $x\sim y$.
We define the \emph{normalized solution space} 
by \[V(I) = V'(I)/ \sim.\]
    
It is not a priori clear that $V(I)$ can be interpreted
as a subset of $\R^{kn}$. And for some instances $I$,
this will not be the case. However,
for the instances that we produce it is.
For us every instance $I$ has a set $S$ of
$k$ disjoint segments associated to it.
We will show that on each segment must
lie exactly one vertex in any valid solution.
Let $\prec$ be an order on $S$. 
Then we can think of $V(I)$ simply as the
vertices of the intermediate polytope given 
in the order $\prec$ and thus $V(I)\subset \R^{kn}$.

\paragraph{Rational-Equivalence.}
On a very high-level, a universality theorem states that we can represent
\emph{any} objects of type~$A$ by an object of type~$B$ 
preserving property~$C$.
In our case, objects of type~$A$ are just bounded algebraic varietes,
which we will formally define below. 
They are very versatile as they can 
encode many different mathematical objects
of interest in a straight-forward fashion.
Instances of the \NPPshort/ are the objects of type~$B$. 
At last, we want to preserve algebraic and topological 
properties. 
In this paragraph, we define the notion of rational-equivalence,
which preserve both~\cite{shitov2016universality}.

Let $F$ be a finite set of polynomials $F = \{f_1,\ldots,f_k\}$ 
with $f_i\in\Z[x_1,\ldots,x_n],\,  i=1,\ldots,k$.
Then we define the variety of $F$ as 
\[V(F) = \{x\in \R^n : f(x) = 0, \ \forall f\in F\}.\]
We say $V(F)$ is bounded, if there is a ball $B$ such that
$V(F)\subseteq B$.

Two varieties $V$ and $W$ are \emph{rationally equivalent} 
if there exists a homeomorphism
$f : V \rightarrow W$ 
such that both $f$ and $f^{-1}$ are given by rational functions.
A function $f:V\rightarrow W$ is a homeomorphism, if it is
continuous, invertible and its inverse is continuous as well.
The function $f$ is rational, if it can be component-wise described
as the ratio of polynomials. 
We denote rational-equivalence by $V\simeq W$.
Note that the composition of two homeomorphisms
is a homeomorphism. Similarly, the composition of 
two rational functions is rational.
Next to algebraic and topological properties, we preserve also
algorithmic properties. To state this properly, we will
introduce the complexity class \ER/ in the next paragraph.

\paragraph*{Existential Theory of the Reals.}
In the study of geometric problems, the complexity class \ER/
plays a  crucial role, connecting purely geometric problems 
and Real Algebraic Geometry.
Whereas NP is defined in terms of existentially quantified 
Boolean variables, \ER/ 
deals with existentially quantified real variables.  

Consider a 
first-order formula over the reals that contains only 
existential quantifiers,
 \[\exists x_1,x_2,\ldots,x_n :\Phi(x_1,x_2,\ldots,x_n),\] where $x_1,x_2,\ldots,x_n$ are 
 real-valued variables and $\Phi$ is a quantifier-free formula 
 involving equalities and inequalities of integer polynomials. 
The algorithmic problem \textsc{Existential Theory of the Reals} (\ETR/) takes
such a formula as an input and asks whether it is satisfiable. 
The complexity class \ER/ consists of all problems that reduce 
in polynomial time to \ETR/. 
Many problems in combinatorial geometry and geometric 
graph representation naturally lie in this class, and furthermore,
many have been shown  
to be \ER/-complete, e.g., stretchability of a pseudoline
arrangement~\cite{matousekSegments,Mnev,SchaeferS17},  
recognition of segment intersection graphs~\cite{KratochvilM94} 
and disk intersection graphs~\cite{mcdiarmid2013integer}, 
computing the rectilinear crossing number 
of a graph~\cite{Bienstock91}, etc.
For surveys on \ER/, 
see~\cite{Schaefer09,cardinal2015computational,matousekSegments}.
A recent proof that 
the \AGP/ is \ER/-complete~\cite{ARTETR} provides 
the framework we follow in our~proof.
See also~\cite{SymmetricNash,shor1991stretchability, cardinal2015computational,SchaeferS17, Schaefer09, matousekSegments,AreaKleist, lubiw2018complexity, mcdiarmid2013integer,richter1995realization, cardinal2017recognition, cardinal2017intersection}
for a small selection of \ER/-complete problems.

\paragraph{Results.}
We show a universality theorem for the \NPPshort/. Note that
the result is implied by a  recent result of 
Shitov~\cite{shitov2016universality} about \NMF/
and an old reduction due to 
Cohen and Rotblum~\cite{cohen1993nonnegative}.
Thus we attribute the result to Shitov.
\begin{theorem}[Universality Shitov~\cite{shitov2016universality}.]
\label{thm:Universality}
    For every bounded variety $V(F)$ 
    exists an instances $I$
    of the \NPP/ such that $V(I) \simeq V(F)$.
\end{theorem}
In this paper, we give a direct proof that does not use 
either of the two above papers. The main ideas of our proof 
are simple geometric constructions. 
This implies that polynomial equations  
have a solution space that is topologically 
and algebraically equivalent to solution spaces given 
by the \NPPshort/.
To illustrate the strength of the statement,
we highlight give one algebraic corollary and 
one topological example.
\begin{corollary}[Algebraic Consequences]
    Let $\Q \subseteq F_1\subset F_2 \subset \R$ be two algebraic field extensions of~\Q. Then there exists an instance of the \NPPshort/
    such that there is a solution in~$F_2$, but not in~$F_1$.
\end{corollary}
This implies for instance the result by~\cite{chistikov2017nonnegative},
who showed that there is an instance of the \NPP/ that requires 
irrational coordinates.

\begin{example}[Topological Consequences]
    Let $T$ be a torus, then there is an instance
    of the \NPPshort/ such that the solution space
    is homeomorphic to $T$.
\end{example}
Note that the polynomial equation
\[f(x,y,z) = (x^2+y^2+z^2+R^2-r^2)^2-4R^2(x^2+y^2)=0\]
describes a torus with the two radii $r$ and $R$.
To see the last corollary, 
simply apply Theorem~\ref{thm:Universality}
on the variety given by $f$, with $R = 10, r = 1$.

As all the steps involved to show this universality theorem
take polynomial time to execute we can infer the algorithmic
complexity of the \NPPshort/.
\begin{corollary}[Shitov~\cite{shitov2016universality}]
    The \NPP/ is \ER/-complete.
\end{corollary}

In the rest of the introduction, we survey the literature
on the \NPPshort/ and the closely related problem of \NMF/.

\paragraph{Proof Overview.}
The proof consists of two parts.
In Section~\ref{sec:EncodingETR}, we show that
a certain very simple set of polynomial equations
is already \ER/-complete and admits the
desired universality property. In a nutshell,
we only allow only the constraints
$x\cdot y = 1$ and $x+y+z = 5/2$.

In the second part, in Section~\ref{sec:PolytopeBuilding}, 
we are encoding those 
constraints in the \NPPshort/. 
The first idea is to enforce certain
vertices to lie on specific line segments.
Those vertices are encoding variables.
It is very easy to build polytopes that
encode the two constraints explained above.
The main technical challenge is to "stick"
those smaller building blocks together to a "big one".
This is easy, \emph{if} you are used to work with polytopes
in higher dimensions. In our description, we do not
assume the reader to have that familiarity.
See also Figure~\ref{fig:BuildingBlocks}

\begin{figure}
    \centering
    \includegraphics{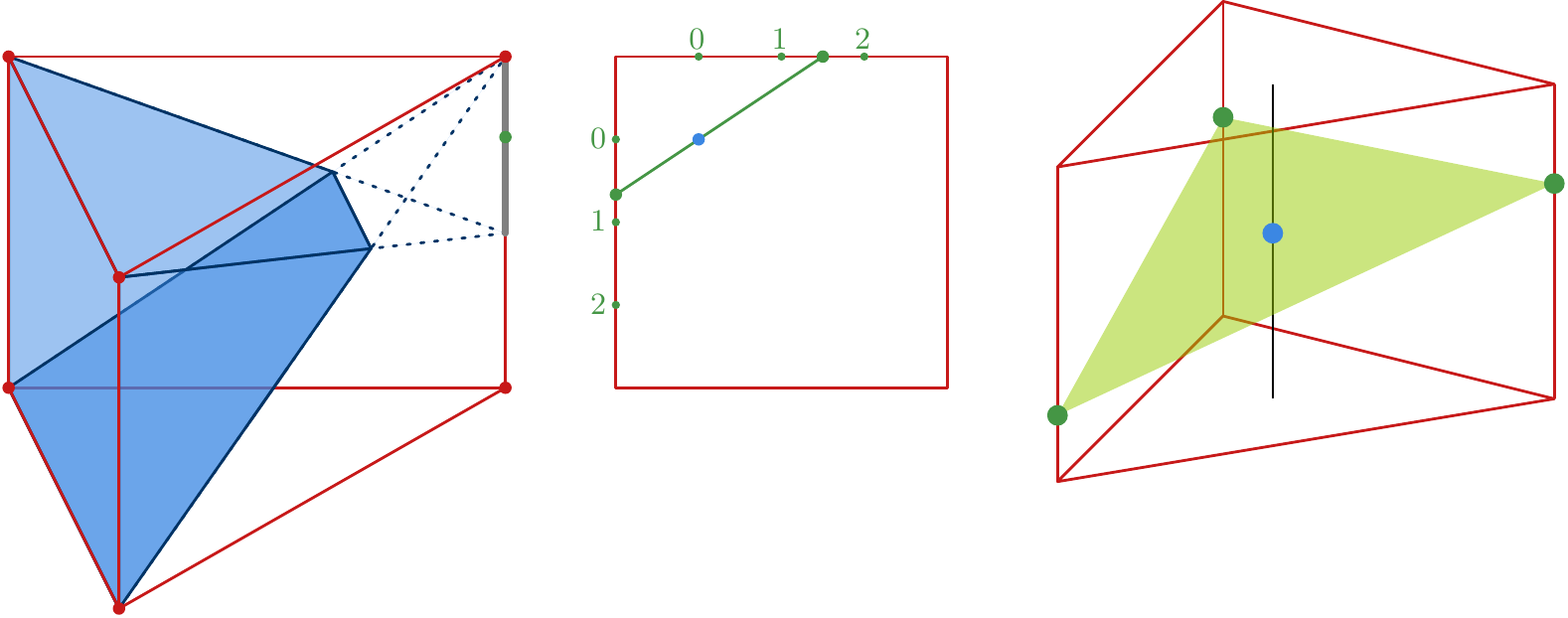}
    \caption{Red indicates the outer polytope, blue the inner polytope, and green the nested polytope. 
    Left: The nested polytope must have one vertex on 
    the left bottom edge. Middle: Assuming that one nested vertex is on the left edge and one at the top edge, then the two 
    vertices encode inversion. Right: Assuming that
    each green vertex is forced to be on its vertical segment,
    then the blue vertex enforces the constraint $x+y+z = 5/2$.}
    \label{fig:BuildingBlocks}
\end{figure}

\paragraph*{Related Work on Nested Polytopes.}
To the best of our knowledge the
\NPPshort/ was first mentioned by Silio in 1979~\cite{silio1979efficient},
who could find an $O(nm)$ time algorithm in the
case that the outer and inner polytope 
are convex polygons in the plane with
$n$ and $m$ vertices respectively.
Additionally, Silio restricts to the case $k=3$.
The motivation of Silio came from a 
connection to Stochastic Sequential Machines.

Independently, Victor Klee suggested the same problem
as was pointed out in several papers~\cite{dasCCCG, aggarwal1989finding, das1990approximation, o1988computational,das1997complexity},
the first of them dating back to 1985.
In particular, the \NPPshort/ appears as one of the open problem 
in the Computational Geometry Column $\# 4$~\cite{o1988computational}.
The main motivation of those early papers 
used to be simplification of a given 
polytope, see Figure~\ref{fig:Simplification}.
\begin{figure}
    \centering
    \includegraphics{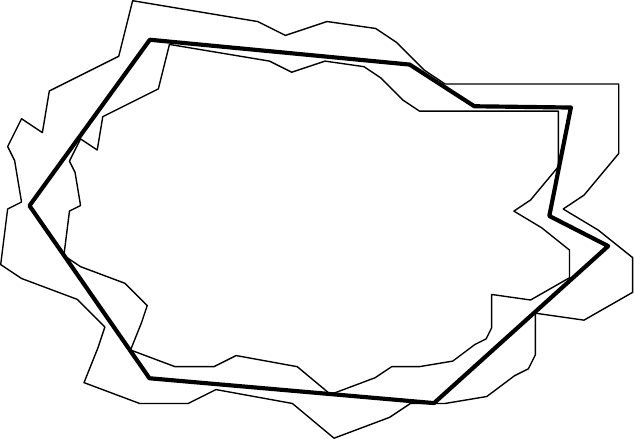}
    \caption{Using an algorithm for the nested polygon problem
    it is possible to attain a simplified version of the previous
    polygon.}
    \label{fig:Simplification}
\end{figure}

Among the first results is an $O(n\log k)$ algorithm for the nested
convex polygon problem~\cite{aggarwal1989finding}.
On the lower bounds side, Das and Joseph showed NP-hardness
for the \NPPshort/ in dimension three~\cite{dasCCCG, das1990approximation,das1992minimum,das1997complexity}. 

In 1995, Suri and Mitchell were able to reduce the \NPPshort/ to a 
set cover problem, by loosing only a factor of $d$.
Their motivation to study the \NPPshort/ came from separating
geometric objects.
Using the greedy approximation scheme for set cover
they attain an $O(d\log n)$-approximation algorithm that runs
in $O(n^{d+1})$ time 
($n=$ number of facets of inner and outer polytope, $d=$ dimension).
This was consequently improved by 
Br{\"o}nnimann and Goodrich~\cite{bronnimann1995almost}
as the first application in their seminal paper on $\varepsilon$-nets.
The key observation is that the set-cover system
described by Suri and Mitchell has bounded VC-dimension.
Their algorithm runs in $O(n^{d+2}\log^d n)$ and gives an
$O(d^2 \log OPT)$-approximation ($OPT =$ size of the optimal solution) .
Independently Clarkson~\cite{clarkson1993algorithms} found a
similar approximation algorithm using techniques 
from linear programming.

Interestingly, the \NPPshort/ has close relations to \NMF/.
We define \NMFshort/, explain the history and this relation
in the next paragraph.

\paragraph{Non-Negative Matrix Factorization.}
In a parallel line of research the \NMFshort/ 
is explored, with the earliest mentioning,
we found, in 1973~\cite{berman1973rank}.
The non-negative matrix factorization is defined as follows. 
Given a matrix $M\in \R_+^{m\times n}$ and a number $k$,
we say that $M = V \cdot W$, for matrices $V\in \R_+^{m\times k}$
and $W\in \R_+^{k\times n}$, is a non-negative matrix factorization
of inner dimension $k$. We denote by $\R_+$ the set of 
non-negative real numbers.
We denote by $\rank_+(M)$ the non-negative rank, which 
is the smallest inner dimension, for which a non-negative
matrix factorization exists.

While it is said that \NMFshort/ has many applications in image processing,
machine learning, dimension reduction and clustering,
in theory it is most famous for the relationship to extension complexity.
Given a polytope $P$, its extension complexity, is the smallest
number $k$ such that there is a polytope $Q$ on $k$ facets
such that there is a linear projection from $Q$ to $P$.
Yannakakis showed in its seminal paper~\cite{yannakakis1991expressing} 
(Roughly saying that on needs an exponential size \emph{symmetric} LP
to solve the travelling salesperson problem.)
that the extension complexity of a polytope is
the non-negative rank of its slack matrix.
See also~\cite{SamNonSymmetric} for the lower bound for
non-symmetric LPs.

For us most relevant is a reduction from the \NMFshort/ to \NPPshort/
by Cohen and Rothblum~\cite{cohen1993nonnegative}.
To the delight of the reader, we repeat this reduction in the appendix.
\begin{restatable}[Cohen Rothblum~\cite{cohen1993nonnegative}]
{lem}{CohenReduction}
\label{lem:CohenRothblum}
    Let  $M\in \R^{m\times n}_+$ be a matrix and $k\in \N$ be a number.
    Let $A$ be the convex hull intersected with 
    $H = \{x\in \R^{m} : \sum_i x_i = 1\}$ and
    $B$ the positive orthant intersected with $H$.
    Then $(A,B,k)$ as an instance to the \NPPshort/ is equivalent
    to $(M,k)$ as \NMFshort/.
\end{restatable}

In 2009, Vavasis~\cite{vavasis2009complexity} showed
that exact-non-negative matrix factorization is equivalent 
to the intermediate simplex problem.
In particular, their results imply that 
the \NPPshort/ is already hard, if the intermediate 
polytope is restricted to be a simplex.
In a similar way,
Gillis and Glineur~\cite{gillis2012geometric} in~2010
showed that restricted non-negative matrix factorization
is equivalent to \NPPshort/.

Although, it is easy to encode \NMFshort/ as an algebraic decision
problem the huge number of variables makes it algorithmically 
infeasible. 
In 2012, Arora Ge, Kannan and Moitra~\cite{AroraNMFProvably} found an
algorithm that runs in polynomial time for every fixed $k$.
They also showed that there is
no $(nm)^{o(k)}$, assuming ETH.
In 2016 Moitra~\cite{MoitraAlmostOpt} improved 
the upper bound and gave an $(nm)^{O(k^2)}$ algorithm for \NMFshort/.
Note that those results translate immediately
to results about \NPPshort/, due to the reductions mentioned above.

Chistikov, Kiefer, Maru{\v{s}}i{\'c}, Shirmohammadi, Worrell
have shown that the optimal solution of 
the $\textsc{Nested}$ $\textsc{Polytope Problem}$ requires irrational coordinates already
in dimension $d=3$ and $k=5$. It is an open problem,
if the intermediate-simplex problem requires irrational
coordinates as well.
This was also shown in parallel by 
Shitov~\cite{shitov2017nonnegative}.
Furthermore, Shitov showed a universality
result for \NMFshort/, very similar to our result. 
In particular, his result implies \ER/-completeness
of both \NMFshort/ and the \NPPshort/.

\section{Encoding ETR}
\label{sec:EncodingETR}
In this section, we define the algorithmic
problem and the complexity class both called
the \textsc{Exitential Theory of the Reals}.
For distinction, the algorithmic problem is denoted 
by \ETR/ and the complexity class by \ER/.

An instance of \ETR/ is a well-formed logical formula
of the form \[\exists x_1,\ldots,x_n : \phi(x_1,\ldots,x_n).\]
The subformula $\phi$ is quantifier free. 
It has polynomial equations and  (strict) 
inequalities as atomic formulas.
Those atomic formulas can be combined in any boolean way.
For example:
\[\exists x,y,z : [(x^2 + y^2 = 1) \land (x = -2)]
\lor \lnot (y^2z < -1).\]
Strictly speaking, we are only allowed to use variables and 
the symbols
\[\Sigma = \{+,\cdot, =,>,\leq,(,),0,1,\land,\lor,\lnot\}.\]
However, we interpret $x^2$ as $x\cdot x$ and $x = -2$
as $x + 1 + 1 = 0$ and so on.
We are asking if there is an assignment of real numbers to
the variables such that the formula $\phi$ becomes true.
Note that the definition in the introduction and the more 
precise definition here are equivalent, although this
may not be obvious.

It is not a priory clear that there even exists an algorithm
which can decide this problem. Due to Tarski's 
Quantifier Elimination~\cite{Tarski}, we know that this question 
can be decided. Even more, we know that the problem can be solved
in polynomial space, due to Canny~\cite{CannyPSPACE}.
The complexity class \ER/ is defined as the set of algorithmic
problems that can be reduced in polynomial time to
\ETR/.

\bigskip

An {\bf \ETRINVS/} $\mathcal{S}$ of size $n$ is a vector of $n$ real variables $(x_1, \dots, x_n)\in [\frac{1}{2}, 2]^n$ together with a system of linear and quadratic equations of the form
\[x+y = z \;\;\; , \;\;\; x\cdot y = 1. \]

The solution space of an \ETRINVS/ $\mathcal{S}$ is the set of all
vectors in $[\frac{1}{2},2]^n \subset \mathbb{R}^n$ which satisfy the
equations of $\mathcal{S}$, where we do allow the possibility of an 
empty solution space. (Note that the solution space of an \ETRINVS/
 is a real semi-algebraic set.) It was shown 
in \cite[Lemma 12]{ARTETR} that the problem of determining whether 
an \ETRINVS/ has a non-empty solution space is \ER/-complete.
(The original formulation of \ETRINVS/ included the 
equation $x=1$, but this equation can be obtained by 
$x\cdot x = 1$ and  $x\in [\frac{1}{2}, 2]$) 
Although, this was not pointed out directly, if we follow
the reduction it is easy to observe that
all steps are rationally-equivalent.
\begin{lemma}[Universality Inversion]
    Let $F$ be a finite set of polynomials $F = \{f_1,\ldots,f_k\} \subset \Z[x_1,\ldots,x_n]$,
    with bounded solution space.
    There is an instance $I$ of \ETRINVS/ such that
    \[V(I) \simeq V(F).\]
\end{lemma}
\begin{proof}[Proof Sketch.]
    In almost every step of the reduction a new variable and 
    a new constraint is introduce. All other variables are 
    left as they are.
    For example, $Y\cdot X_i - 1 = 0$, where $X_i$ is an old variable and
    $Y$ is a new variable. Of course, there is the assumption that $X\not = 0$.
    We see that $Y = 1/X$, which determines 
    that the new and the old system of equations are rationally equivalent.
    To be explicit, the following mapping 
    \[f: (x_1,\ldots,x_n) \mapsto (x_1,\ldots,x_n,1/x_i),\]
    is a homeomorhpism, 
    it is rational and its inverse is rational as well.
    
    Note that there are two exceptions.
    The inequality $X>0$ is replaced by $XY^2 -1 = 0$.
    Note that this step does not preserve homotopy
    as the two sets 
    \[S = \{x: x>0 \}\]
    and 
    \[T = \{(x,y):  xy^2 -1 = 0\}\]
    do not have the same number of connected components.
    However, we restrict ourselves to systems of polynomial
    \emph{equations}.
    
    The second exception is when all variables are scaled down 
    to a small range. Note that the sets 
    \[A = \{(x,y): x+y =1 \}\]
    and 
    \[B = \{(x,y):  x+y =1 , -L\leq x,y\leq L\}\]
    are not rationally equivalent. 
    But again, this does not apply to us, as we assume that
    the initial solution space is bounded.
\end{proof}

While \ETRINVS/ is the right intermediate problem to show 
that the Art Gallery Problem is \ER/-complete, 
for the purpose of this paper however, 
it will be more convenient 
to work with a slight modification of \ETRINVS/,
which we introduce now. 

An {\bf \ETRINVA/} $\mathcal{A}$ of size $m\times n$ is an $m$-by-$n$ matrix of variables $A = (\alpha_{i,j}) \in [\frac{1}{2}, 2]^{m\times n}$ together with a system of linear and quadratic equations of the form
\[\alpha_{i,j} + \alpha_{i,k} = \tfrac{5}{2} \;\; \; , \;\;\; \alpha_{i,j} + \alpha_{i,k} + \alpha_{i,l} = \tfrac{5}{2} \;\;\; , \;\;\; \alpha_{i,k}\cdot \alpha_{j,k} = 1.\]
(Note that the linear equations relate variables in the same row and the quadratic equations relate variables in the same column.)

The solution space of an \ETRINVA/ is defined similarly as for \ETRINVS/ and is a semi-algebraic subset of $[\frac{1}{2},2]^{m\times n} \subset \mathbb{R}^{m\cdot n}$. We now have the following lemma.

\begin{lemma}[Uninversality of \ETRINVA/]\label{lem:sys-arr}
Let $\mathcal{S}$ be an {\em \ETRINVS/} on $n$ variables.
There exists an {\em \ETRINVA/} \A{} of size $3\times2n$ such that the solution spaces of \S{} and \A{} are rationally-equivalent.
The description complexity of \A{} is linear in \S{}.
\end{lemma}
\begin{proof}
    Let us denote by $x_1,\ldots, x_n$ the variables 
    of \S{}. 
    First note that we can assume without 
    loss of generality that every variable $x_i$
    in \S{} is in at most one inversion constraint 
    involved. Otherwise $x\cdot y =1$ and $y\cdot z =1$,
    implies $x=z$ and we can replace $z$ everywhere by
    $x$ and forget about $z$. Note that this preserves
    rational equivalence.

    We denote the variables of \A{} by
    $y_i,\alpha_i,\beta_i,
    \gamma_i, \delta_i,\varepsilon_i$, for $i=,\ldots,n$
    and we write them into the array as follows
\[\left(
\begin{array}{cc}
     y_1 \ldots y_n &  \alpha_1 \ldots \alpha_n\\
     \beta_1 \ldots \beta_n &   \gamma_1 \ldots \gamma_n \\
     \delta_1 \ldots \delta_n &   
     \varepsilon_1 \ldots \varepsilon_n 
\end{array}
\right).
\]
We want that all constraints in \S{} for the
$x_i$ variables hold in \A{} for the corresponding
$y_i$ variables. The remaining variables in 
\A{} are supposed to be completely determined.
We introduce the linear constraint
\[y_i + \alpha_i = 5/2,\]
for every $i=1,\ldots,n$.
Let us first consider linear constraints of the form
\[x_i+x_j = x_k,\]
in \S{}.
We introduce the linear constraint,
\[y_i+y_j + \alpha_k = 5/2.\]
Note that this implies 
\[y_i + y_j = y_k.\]
This encodes all linear constraints.
Now let us consider the quadratic constraints,
involving two different variables.
Note that we denote the pairs of 
constraints as 
\[C = \{(i,j) : i<j,\, x_i\cdot x_j = 1\}.\]
For every $(i,j) \in C$, we are adding the constraints
\[y_i \cdot \beta_i = 1 \text{ and } \beta_i \cdot \delta_i = 1.\]
Similarly, we add the constraints
\[y_j \cdot \delta_j = 1 \text{ and } \beta_j \cdot \delta_j = 1.\]
This enforces $\beta_i = 1/y_i$ and $\beta_j = y_j$.
Furthermore, we add the constraints 
\[\beta_i + \gamma_i = 5/2,
\text{ and }
\beta_j + \gamma_i = 5/2,\]
This enforces $1/y_i = \beta_i = \beta_j = y_j$, as desired.

Let us now consider the special case of $x_i \cdot x_i = 1$.
Recall that this is equivalent to $x_i =1$.
We introduce the constraints 
\[y_i \cdot \beta_i = 1, \  \beta_i \cdot \delta_i = 1, 
\text{ and } y_i \cdot \delta_i = 1.\]
This is equivalent to  $y_i = 1$.

Note that all the $\varepsilon_i$'s variables and some of the 
$\beta_i, \gamma_i,\delta_i$ variables are 
still completely unconstrained.
Add any linear constraint of the form $a+b = 5/2$ 
with an already used variable 
to them, 
so that they are uniquely determined, 
by one of the $y_i$.
We have that the set of constraints on the 
$x_i$'s and one the $y_i$'s are exactly the same.

Let use denote by $V = V(\S)$ the solution space of \S{} and
by $W = V(\A)$ the solution space of \A{}.
We have to show that $V$ and $W$ are rationally-equivalent.
To this end we define the mapping \[f:V\rightarrow W.\]
Let $(x_1,\ldots,x_n)$ in $V$. Then 
we define $y_i := x_i$. We define
$\alpha_i := 5/2 - x_i$.
Each $\beta_i$ and $\delta_i$ is either $x_i$ or $1/x_i$,
if the index $i$ was contained in a pair in $C$.
All other variables $\theta$ are of the form 
$5/2 - x_j$ or $5/2 - 1/x_j$, for some $j$.
Note that $f$ is bijective, as the set of constraints
onto the $x_i$'s and $y_i$'s are the same.
The mapping $f$ is continuous and rational by definition.
The inverse mapping $f^{-1}: W \rightarrow V$
is simply given by

\[\left(
\begin{array}{cc}
     y_1 \ldots y_n &  \alpha_1 \ldots \alpha_n\\
     \beta_1 \ldots \beta_n &   \gamma_1 \ldots \gamma_n \\
     \delta_1 \ldots \delta_n &   
     \varepsilon_1 \ldots \varepsilon_n 
\end{array}
\right)
\mapsto (y_1,\ldots,y_n).
\]
This is also a continuous and rational mapping.
Note that the number of constraints and variables of \A{} is
linear in the number of variables and constraints of \S{}.
This finishes the proof.
\end{proof}

\section{Building the polytopes}
\label{sec:PolytopeBuilding}

The main goal of this section is to prove the following lemma.
\begin{lemma}\label{lem:nestedLem}
Let $\mathcal{A}$ be an {\em ETR-INV}-array of size $m\times n$. There exists convex polytopes $A\subset B \subset \mathbb{R}^{2+n+m}$ such that there exists a nested polytope $A\subset X \subset B$ with $k=mn+2m+2$ vertices such that
the solution spaces are rationally-equivalent.
\end{lemma}

\begin{remark} The polytopes $A\subset B$ in Lemma \ref{lem:nestedLem} are actually contained in a hyperplane in $\mathbb{R}^{2+n+m}$, and are $(n+m+1)$-dimensional. The outer polytope $B$ will have $(n+2)(m+1)$ vertices and is defined by $n+m+3$ hyperplanes. The vertex description and facet description of the outer polytope will be given in Subsection  \ref{coord outer poly}.

The inner polytope $A$ has $2m+2$ vertices in common with the outer polytope $B$ together with an additional $2mn$ vertices that lie on certain 2-faces of the outer polytope. Finally, for each equation in the ETR-INV-array $\mathcal{A}$ we add one additional vertex to the inner polytope which will lie on certain faces of the outer polytope. The vertex description of the inner polytope $A$ is given in Subsection \ref{coord inner poly}.
\end{remark}

\subsection{Two geometric observations}

Here we state two simple geometric observations that are used for the ``gadgets'' needed in our construction of the polytopes of 
Lemma~\ref{lem:nestedLem}.

\subsubsection{The linear equations}
Let $\{v_0, v_1, \dots, v_k\}$ be a set of affinely independent points in $\mathbb{R}^d$. For $1\leq i \leq k$ let $w_i = v_i+v_0$ and define the prism $P$ as \[P= \conv(\{v_1, \dots, v_k, w_1, \dots, w_k\}).\]
For $t\in [0,1]$ define the point $q_t\in P$ as
\[q_t  = (1-t)(\tfrac{1}{k}v_1 + \cdots + \tfrac{1}{k}v_k) + t(\tfrac{1}{k}w_1 + \cdots + \tfrac{1}{k}w_k) 
= \tfrac{1}{k}v_1 + \cdots + \tfrac{1}{k}v_k + tv_0.\]
Finally, for $1\leq i \leq k$ define points $p_i$ as \[p_i = (1-\lambda_i)v_i + \lambda_iw_i = v_i + \lambda_i v_0,\] where $\lambda_i\in [0,1]$. A simple calculation (left to the reader) gives us the following.

\begin{obs} \label{linear}
$q_t\in \conv(\{p_1, \dots, p_k\})$ if and only if $\sum_{i=1}^k\lambda_i = tk$.
\end{obs}

\subsubsection{The quadratic equation} In the plane $\mathbb{R}^2$, let $p_1 = (\alpha_1,-1)$ be a point on the line $y=-1$ and let $p_2 = (-1,\alpha_2)$ be a point on the line $x=-1$, where $\alpha_1, \alpha_2\in [\frac{1}{2}, 2]$. A simple calculation (left to the reader) gives us the following.

\begin{obs} \label{quadratic}
The origin $(0,0) \in \conv(\{p_1,p_2\})$ if and only if $\alpha_1\cdot \alpha_2 = 1$.
\end{obs}

\subsection{A basic outline of the construction}
We now give an outline of the construction of the polytopes in Lemma \ref{lem:nestedLem}, without giving explicit coordinates, but rather focusing on the three ``gadgets'' that will be used to encode the three types of equations in $\mathcal{A}$. (We will give precise coordinates in Subsection \ref{sect:explicit}.)

\subsubsection{The outer polytope}
To build the outer polytope $B$ we start with an {\em ``orthogonal frame''} spanning $\mathbb{R}^m$, 
consisting of $m$ mutually orthogonal segments of equal length all meeting in a common endpoint. 
Note that the convex hull of these segments form an $m$-dimensional simplex. 
(When we eventually add coordinates, the length of these segments will be 3 units, 
each one parametrizing the closed interval $[-1, 2]$.) We now take $n+2$ distinct 
copies of the orthogonal frame, $U_1$, $U_2$, $V_1, \dots, V_n$, each one translated 
into ``independent dimensions'' so that their union now lives 
in $\mathbb{R}^{2+n+m}$ (Note that the affine span of the union will be $(n+m+1)$-dimensional.) 
We label the segments of these orthogonal frames as
\[
\begin{array}{ccc}
U_1 &= & \{ \tau_{1,1} , \dots,  \tau_{m,1}   \} \\
U_2 &= & \{ \tau_{1,2} , \dots,  \tau_{m,2}   \} \\
V_1 &= & \{ \sigma_{1,1} , \dots,  \sigma_{m,1}   \} \\
& \vdots & \\
V_n &= & \{ \sigma_{1,n} , \dots,  \sigma_{m,n}   \} 
\end{array}
\]
such that the segments $\tau_{i,1}, \tau_{i,2}, \sigma_{i,1}, \dots, \sigma_{i,n}$ are all parallel. 

We now take the outer polytope $B$ to be the convex hull of $U_1 \cup U_2 \cup V_1 \cup \cdots \cup V_n$. It is straight-forward to show that $B$ is an $n+m+1$-dimensional polytope with $(n+2)(m+1)$ vertices.  
In what follows, for each $1\leq i \leq m$ and $1\leq j \leq n$, 
the ``second half'' of  segment $\sigma_{i,j}$, parametrizing the interval $[\frac{1}{2},2]$,
will correspond to the variable $\alpha_{i,j}$ in the ETR-INV-array $\mathcal{A}$. The segments $\tau_{i,j}$ will play an auxiliary role which we describe next. 

\subsubsection{Building the inner polytope: Enforcing vertices to segments}

The first step in building the inner polytope $A$ is to enforce the following. 

\begin{prpty} \label{seg-restr}
Let $X$ be a nested polytope, with $k = mn+2m+2$ vertices and $A\subset X\subset B$. 
For every $1\leq i \leq m$ and $1\leq j \leq n$, 
the segment $\sigma_{i,j}\in V_j$ contains exactly one vertex  
of $X$, which we denote by $x_{i,j}$. 
\end{prpty}
(More specifically, each segment of the orthogonal frame $V_i$ will contain exactly one vertex 
from $X$ in its ``second half'', thus encoding a value in the 
interval $[\frac{1}{2}, 2]$.) This can be done as follows. 
Fix indices $1\leq i \leq m$ and $1\leq j \leq n$, and 
consider segment $\tau_{i,1}\in U_1$ and its parallel 
copy  $\sigma_{i,j}\in V_j$, which are edges of a $2$-dimensional 
face of the outer polytope $B$. Define the point $y_{i,j}$ to be 
the unique point in this 2-face such that segment  
$\tau_{i,1}\in U_1$ is mapped to the second half of its parallel 
copy $\sigma_{i,j}\in V_j$ by central projection through 
$y_{i,j}$. Similarly, we define the analogous point $z_{i,j}$ 
in the $2$-face of $A$ spanned by the segment  $\tau_{i,2}\in U_2$ 
and its parallel copy $\sigma_{i,j} \in V_j$. (See Figure \ref{fig:ForceSegment}.)

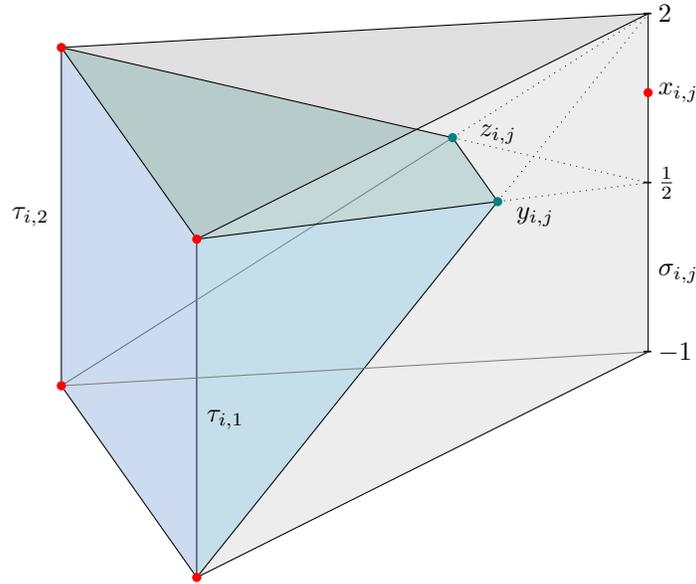
\begin{figure}[ht]
    \centering
\begin{tikzpicture}[scale=1.5]

    \begin{scope}
    \draw (0,-0.5) coordinate(A) -- (0,2.5) coordinate (B);
    \draw (4,1.5) coordinate(C) -- (4,3) coordinate (G) -- (4,4.5) coordinate (D);
    \draw (-1.2,1.2) coordinate(E) -- (-1.2,4.2) coordinate (F);

    \draw[dotted] (intersection of  E--D and F--G) -- (D);
    \draw[dotted] (intersection of  E--D and F--G) -- (G);
    \draw[dotted] (intersection of  A--D and B--G) -- (D);
    \draw[dotted] (intersection of  A--D and B--G) -- (G);

    \draw[white!50!black] (E) -- (intersection of  E--D and F--G);
    \fill[cyan, opacity = .2] (A) -- (B) -- (intersection of  A--D and B--G);
    \fill[teal, opacity = .2] (B) -- (F) -- (intersection of  E--D and F--G) -- (intersection of  A--D and B--G);
    \fill[blue!40!teal, opacity =.2] (A) -- (B) -- (F) --(E);
    \fill[black!60!white, opacity =.2] (B) -- (D) -- (F);
    \fill[black!35!white, opacity =.2] (A) -- (B) -- (D) -- (C);

    \draw (A) -- (intersection of  A--D and B--G) -- (intersection of  E--D and F--G) -- (F) -- (B) -- (intersection of  A--D and B--G); 

    \draw (A) --(E);
    \draw (A) -- (C); 
    \draw[gray] (E) -- (C); 
    \draw (F) -- (D);
    \draw (B) -- (D);
    
    \node[teal,scale=3] at (intersection of  E--D and F--G){.};

    \node[teal,scale=3] at (intersection of  A--D and B--G){.};

\node[red, scale=3] at (4,3.8) {.};

\node[red, scale=3] at (0,-.5) {.};

\node[red, scale=3] at (0,2.5) {.};

\node[red, scale=3] at (-1.2,1.2) {.};

\node[red, scale=3] at (-1.2,4.2) {.};

\node[right] at (4,3.8) {\small $x_{i,j}$};

\node[right] at (0,0.9) {\small $\tau_{i,1}$};
\node[left] at (-1.22,2.7) {\small $\tau_{i,2}$};
\node[right] at (4.0,2.2) {\small $\sigma_{i,j}$};
\node at (3,2.7) {\small $y_{i,j}$};
\node at (2.67,3.43) {\small $z_{i,j}$};

\node at (C) {-};
\node at (G) {-};
\node at (D) {-};

\node[right] at (C) {\small $-1$};
\node[right] at (G) {\small $\frac{1}{2}$};
\node[right] at (D) {\small $2$};

\end{scope}

\end{tikzpicture}
\caption{The vertices of any nested polytope $A\subset X\subset B$ (marked in red) must include the endpoints of segments $\tau_{i,1} \in U_1$ and $\tau_{i,2}\in U_2$, while the last vertex, $x_{i,j}$, must be contained in the segment  $\sigma_{i,j}\in V_j$ restricted to the interval $[\frac{1}{2}, 2]$. }
\label{fig:ForceSegment}
\end{figure}

At this stage of the construction the inner polytope $A$ will consist of the orthogonal frames $U_1$ and $U_2$ together with the points $\{y_{i,j}, z_{i,j}\}$ for all $1\leq i \leq m$ and $1\leq j \leq n$. Moreover, if $X$ is a nested polytope, with $mn+2m+2$ vertices and $A\subset X\subset B$, then $X$ must contain the orthogonal frames $U_1$ and $U_2$.
This accounts for $2m+2$ of the vertices, as all $\tau_{i,1}$ (the same for the $\tau_{i,2}$.) have one end point in common.
Futermore, $X$ must contain one vertex in each of the segments of the orthogonal frames $V_1, \dots, V_n$.
This accounts for the remaining $m\cdot n$ vertices. 
Thus Property \ref{seg-restr} is satisfied, and we let $x_{i,j}$ 
denote the unique vertex of $X$ which is contained in the (second half of the) 
segment $\sigma_{i,j}\in V_j$, which we associate with  
the variable $\alpha_{i,j} \in [\frac{1}{2},2]$.

\subsubsection{Building the inner polytope: Encoding $\alpha_{i, j} + \alpha_{i,k} = \frac{5}{2}$} 

In order to enforce the relation $\alpha_{i,j} + \alpha_{i,k} = \frac{5}{2}$, we add a new vertex $p_{i,j,k}$ to the inner polytope $A$ as follows. 
We consider the rectangular 2-face of the outer polytope $B$ spannced by the segements $\sigma_{i,j}\in V_j$ and $\sigma_{i,k}\in V_k$. Define $p_{i,j,k}$ to be the point in this 2-face such that $p_{i,j,k}$ is contained in the convex hull of the vertices $x_{i,j}$ and $x_{i,k}$ of the nested polytope $X$ (satisfying Property \ref{seg-restr}) if and only if the associated variables $\alpha_{i,j} + \alpha_{i,k} = \frac{5}{2}$. (The unique point $p_{i,j,k}$ exists by Observation \ref{linear} by letting $\{v_1,w_1\}$ be the endpoints of $\sigma_{i,j}$ and $\{v_2,w_2\}$ be the endpoints of $\sigma_{i,k}$. See Figure \ref{fig:2}.)

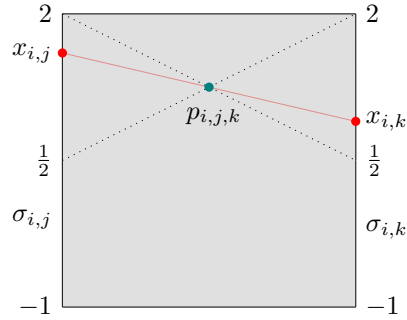
\begin{figure}[ht]
    \centering
    \begin{tikzpicture}[scale =1.3]
        \begin{scope}
            \draw (0,-1) coordinate (A) -- (0,0.5) coordinate (B);
            \draw (B) -- (0,2) coordinate (C);
            \draw (3,-1) coordinate (D) -- (3,0.5) coordinate (E);
            \draw (E) -- (3,2) coordinate (F);
            \draw (A) -- (D);
            \draw (C) -- (F);
            \draw[red, opacity = .4] (0,1.6) -- (3,0.9);
            \fill[black!60!white, opacity =.2] (A) -- (C) -- (F) -- (D);
            \draw[thin, dotted] (B) -- (F) (C) -- (E);
            \node[teal,scale=3] at (intersection of  B--F and C--E){.};
            \node[left] at (0,-0.1) {\small $\sigma_{i,j}$};
            \node[red, scale=3] at (3,0.9) {.};
            \node[right] at (3,0.9) {\small $x_{i,k}$};
            
            \node[right] at (3,-0.2) {\small $\sigma_{i,k}$};
            \node[below] at (1.55,1.15) {\small $p_{i,j,k}$};
            \node[red, scale=3] at (0,1.6) {.};
            \node[left] at (0,1.6) {\small $x_{i,j}$};

            \node[left] at (A) {\small $-1$};
            \node[left] at (B) {\small $\frac{1}{2}$};
            \node[left] at (C) {\small $2$};
            \node[right] at (D) {\small $-1$};
            \node[right] at (E) {\small $\frac{1}{2}$};
            \node[right] at (F) {\small $2$};
        \end{scope}
    \end{tikzpicture}
    \caption{The vertices $x_{i,j}$ and $x_{i,k}$ contain the point $p_{i,j,k}$ in their convex hull if and only if the associated variables satisfy the equation $\alpha_{i,j} + \alpha_{i,k} = \frac{5}{2}$}
    \label{fig:2}
\end{figure}

 Let us now add the vertex $p_{i,j,k}$ to $A$, and consider a nested polytope $X$ satisfying Property \ref{seg-restr}. It is  easily shown that the only possible way that $X$ can contain the point $p_{i,j,k}$, is if this point is contained in $\conv(\{x_{i,j}, x_{i,k}\})$. In other words, if $X$ satisfies Property \ref{seg-restr}, then the associated variables satisfy the equation $\alpha_{i,j}  + \alpha_{i,k} = \frac{5}{2}$.

\subsubsection{Building the inner polytope: Encoding $\alpha_{i,j}+\alpha_{i,k} + \alpha_{i,l}= \frac{5}{2}$}

Enforcing the relation $\alpha_{i,j} + \alpha_{i,k} + \alpha_{i,l} = \frac{5}{2}$ is similar to the previous case,  and we add a new vertex $q_{i,j,k,l}$ to the inner polytope $A$ as follows.
We consider the triangluar prism spanned by the segments $\sigma_{i,j}\in V_j$, $\sigma_{i,k}\in V_k$, and $\sigma_{i,l}\in V_l$, which is a 3-face of the outer polytope $B$. 

Define $q_{i,j,k, l}$ to be the point in this 3-face such that $q_{i,j,k, l}$ is contained in the convex hull of the vertices $x_{i,j}$, $x_{i,k}$, and $x_{i,l}$ of the nested polytope $X$ (satisfying Property \ref{seg-restr}) if and only if the associated variables $\alpha_{i,j} + \alpha_{i,k} +\alpha_{i,l} = \frac{5}{2}$. (The unique point $q_{i,j,k, l}$ exists by Observation \ref{linear} by letting $\{v_1,w_1\}$ be the endpoints of $\sigma_{i,j}$,  $\{v_2,w_2\}$ be the endpoints of $\sigma_{i,k}$, and $\{v_3,w_3\}$ be the endpoints of $\sigma_{i,l}$. See Figure \ref{fig:Addition}.)

\begin{figure}[ht]
    \centering
    \begin{tikzpicture}[scale=1.6]
        \begin{scope}
        \draw (0,-1) coordinate (A) -- (0,1/2) coordinate (B);
        \draw (B) -- (0,2) coordinate (C);
        \draw (2.4,-2) coordinate (D) -- (2.4,-1/2) coordinate (E);
        \draw (E) -- (2.4,1) coordinate (F);
        \draw (4,-1) coordinate (G) -- (4,1/2) coordinate (H);
        \draw (H) -- (4,2) coordinate (I);

        \fill[white!30!black, opacity =.2] (A) -- (C) -- (F) -- (D);        
        \fill[white!70!black, opacity =.2] (F) -- (D) -- (G) -- (I);        
        \fill[white!50!black, opacity =.2] (C) -- (F) -- (I);        

        \draw (A) -- (D) -- (G);
        \draw (C) -- (F) -- (I) --cycle;
        \draw[gray] (A) -- (G) ;
        
        \node[red, scale =3] at (0,1.2) {.};
        \node[left] at (0,1.2) {\small $x_{i,j}$};
        
        \node[red, scale =3] at (2.4,0.1) {.};
                \node[right] at (2.4,0) {\small $x_{i,k}$};

        \node[red, scale =3] at (4,0.9) {.};

        \node[right] at (4,0.9) {\small $x_{i,l}$};

        \node[teal, scale=3] at (2,0.7) {.}; 
        
        \fill[red, opacity = .14] (0,1.2)--(2.4,0.1)--(4,0.9)--cycle; 

        \node at (A) {-}; 
        \node at (B) {-}; 
        \node at (C) {-}; 

        \node[left] at (A) {\small $-1$}; 
        \node[left] at (B) {\small $\frac{1}{2}$}; 
        \node[left] at (C) {\small $2$}; 

        \node[right] at (D) {\small $-1$}; 
        \node[right] at (E) {\small $\frac{1}{2}$}; 
        \node[right] at (2.4,0.92) {\small $2$}; 

        \node at (D) {-}; 
        \node at (E) {-}; 
        \node at (F) {-}; 

        \node at (G) {-}; 
        \node at (H) {-}; 
        \node at (I) {-}; 

        \node[right] at (G) {\small $-1$}; 
        \node[right] at (H) {\small $\frac{1}{2}$}; 
        \node[right] at (I) {\small $2$}; 

\node at (1.6,0.76) {\small $q_{i,j,k,l}$};

\node[left] at (0,-.2) {\small $\sigma_{i,j}$};

\node[left] at (2.4,-1.2) {\small $\sigma_{i,k}$};

\node[right] at (4,-.1) {\small $\sigma_{i,l}$};

        \end{scope}
    \end{tikzpicture}
    \caption{The vertices $x_{i,j}$, $x_{i,k}$, and $x_{i,l}$ contain the point $q_{i,j,k,l}$ if and only if the associated variables satisfy the equation $\alpha_{i,j} +\alpha_{i,k}+ \alpha_{i,l} = \frac{5}{2}$.}
    \label{fig:Addition}
\end{figure}
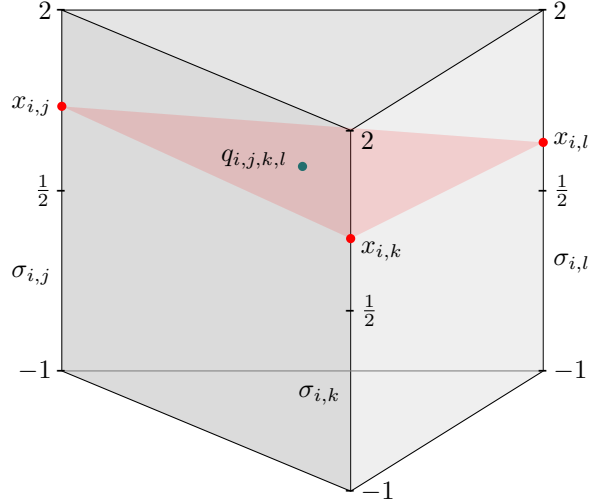

Let us now add the vertex $q_{i,j,k,l}$ to $A$, and consider a nested polytope $X$ satisfying Property \ref{seg-restr}. It is easily seen that the only possible way that $X$ can contain the point $q_{i,j,k,l}$ is if this point is contained in $\conv(\{x_{i,j}, x_{i,k}, x_{i,l}\})$. In other words, if $X$ satisfies Property \ref{seg-restr}, then the associated variables satisfy the equation $\alpha_{i,j}+ \alpha_{i,k} + \alpha_{i,l} = \frac{5}{2}$.

\subsubsection{Building the inner polytope: Encoding $\alpha_{i,k}\cdot \alpha_{j,k}=1$}

In order to enforce the relation $\alpha_{i,k} \cdot \alpha_{j,k} = 1$ we add a new vertex $r_{i,j,k}$ to the inner polytope $A$ as follows. 
Consider the triangular 2-face of $B$ spanned by segments $\sigma_{i,k}\in V_k$ and $\sigma_{j,k}\in V_k$. 
Note that the two segments belong to the same orthogonal frame and thus 
share an endpoint and are orthogonal to one another by definition.
We can coordinatize the plane containing this 2-face such that the segment $\sigma_{i,k}$ 
is parametrized by $\{(x,-1) : -1\leq x \leq 2\}$ and 
the segment $\sigma_{j,k}$ is parametrized by $\{(-1,y) : 1\leq y \leq -2\}$. 
We then define $r_{i,j,k}$ to be the origin with respect to this coordinate system. 
It follows from Observation \ref{quadratic} that the vertices $x_{i,k}$ and $x_{j,k}$ 
contain the point $r_{i,j,k}$ in their convex hull if and only 
if the associated coordinates satisfy the equation 
$\alpha_{i,k}\cdot \alpha_{j,k} = 1$. (See Figure \ref{fig:4}.)

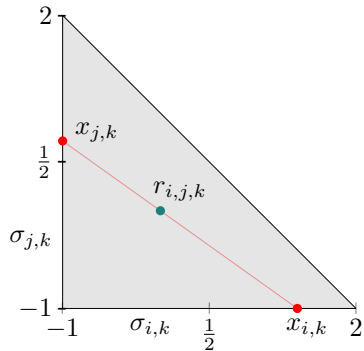
\begin{figure}[ht]
    \centering
    \begin{tikzpicture}
        \begin{scope}[scale=1.3]
            \draw (-1,-1) coordinate(A) -- (2,-1) coordinate(B);
            \draw (-1,-1) coordinate(C) -- (-1,2) coordinate(D);
            \draw (2,-1) -- (-1,2);
            \draw[red, opacity = .4] (-1,5/7) -- (7/5,-1);
            \node[teal, scale=3] at (0,0) {.};
            \fill[white!50!black, opacity=.2] (-1,-1) -- (2,-1) -- (-1,2);
            \node[scale = .4] at (A) {$|$};
            \node[scale = .4] at (1/2,-1) {$|$};
            \node[scale = .4] at (B) {$|$};
            \node[below] at (A) {\small $-1$};
            \node[below] at (1/2,-1) {\small $\frac{1}{2}$};
            \node[below] at (B) {\small $2$};
            \node at (C) {-};
            \node at (-1,1/2) {-};
            \node at (D) {-};
            \node[left] at (C) {\small $-1$};
            \node[left] at (-1,1/2) {\small $\frac{1}{2}$};
            \node[left] at (D) {\small $2$};
            \node[red, scale = 3] at (-1,5/7) {.};
            \node[red, scale = 3] at (7/5,-1) {.};
            \node[above] at (0.2,0) {\small $r_{i,j,k}$};
            \node[above] at (-.65,.6) {\small $x_{j,k}$};
            \node[below] at (1.5,-1) {\small ${x_{i,k}}$};
            \node[left] at (-1,-.3) {\small $\sigma_{j,k}$};
            \node[below] at (-.1
        ,-1) {\small $\sigma_{i,k}$};
        \end{scope}
    \end{tikzpicture}
    \caption{The vertices $x_{i,k}$ and $x_{j,k}$ contain the point $r_{i,j,k}$ if and only if the associated variables satisfy the equation $\alpha_{i,k}\cdot \alpha_{j,k}=1$.}
    \label{fig:4}
\end{figure}

Let us now add the point $r_{i,j,k}$ to $A$, and consider a nested polytope $X$ satisfying Property \ref{seg-restr}. As before, it is easily shown that the only way that $X$ can contain the point $r_{i,j,k}$ is if this point is contained in $\conv(\{x_{i,k}, x_{j,k}\})$. In other words,  if $X$ satisfies Property \ref{seg-restr}, then the associated variables satisfy the equation $\alpha_{i,k}\cdot \alpha_{j,k} = 1$.

\subsection{Explicit coordinates} \label{sect:explicit}

We now give the explicit coordinates to the construction in the previous section. 
Let $\{e_1, e_2,f_1 \dots,f_n,g_1, \dots, g_m\}$ denote the standard basis in $\mathbb{R}^{2+n+m}$ and set \[J=\sum_{j=1}^{m} g_j.\]

\subsubsection{The outer polytope $B$} \label{coord outer poly}

We start by giving the vertices of the outer polytope $B$. First define 
\[u_{0,1} = e_{1}-J =  (1,0,\underbrace{0,\dots,0}_n,\underbrace{-1,\dots,-1}_m),\]
and for $1\leq i \leq m$, let
\[u_{i,1} = e_1 +3g_{i} -J = (1,0,\underbrace{0,\dots,0}_n,\underbrace{-1,\dots, \overset{i}{2}, \dots,-1}_m).\]
This defines the orthogonal frame $U_1 = \{\tau_{1,1}, \dots, \tau_{m,1}\}$ by setting the segment $\tau_{i,1} = \conv(\{u_{0,1}, u_{i,1}\})$.

Similarly, define
\[u_{0,2} = e_{2}-J =  (0,1,\underbrace{0,\dots,0}_n,\underbrace{-1,\dots,-1}_m),\]
and for $1\leq i \leq m$, let
\[u_{i,2} = e_2 +3g_{i} -J = (0,1,\underbrace{0,\dots,0}_n,\underbrace{-1,\dots, \overset{i}{2}, \dots,-1}_m).\]
This defines the orthogonal frame $U_2 = \{\tau_{1,2}, \dots, \tau_{m,2}\}$ by setting the segment $\tau_{i,2} = \conv(\{u_{0,2}, u_{i,2}\})$.

\medskip

Next we define the orthogonal frames $V_1, \dots, V_n$.
For $1\leq j \leq n$, let
\[v_{0,j} = f_j - J =  (0,0,\underbrace{0,\dots,\overset{j}{1},\dots,0}_n,\underbrace{-1,\dots,-1}_m),\]
and for $1\leq i \leq m$, $1\leq j \leq n$, let
\[v_{i,j} = f_j + 3g_{i} - J = 
(0,0,\underbrace{0,\dots,\overset{j}{1},\dots,0}_n,\underbrace{-1,\dots,\overset{i}{2},\dots,-1}_m).\]
For every $1\leq j \leq n$ this defines the orthogonal frame $V_j=\{\sigma_{1,j},\dots, \sigma_{m,j}\}$ by setting $\sigma_{i,j} = \conv(\{v_{0,j}, v_{i,j}\})$.

Finally, for $0\leq i \leq m$, $1\leq j \leq n$, and $k=1,2$ set
$U = \{u_{i,k}\}$ and $V = \{v_{i,j}\}$.
The outer polytope $B$ is now defined as $B = \conv(U \cup V)$. 
Equivalently, $B$ is the set of $x \in \mathbb{R}^{2+n+m}$ in the affine hyperplane 
\[ \left\langle (e_1+e_2+f_1+\dots +f_n), x \right\rangle = 1, \]
that satisfy the following $(2+n+m+1)$ linear constraints,
\[ 
\langle e_i,x \rangle \geq 0, \quad  
\langle f_i,x \rangle \geq 0, \quad  
\langle g_j,x \rangle \geq -1, \quad  
\langle J,x \rangle \leq 3-m.  
\]
Observe that $B$ is a $(n+m+1)$-dimensional convex polytope with vertex set $U\cup V$ and facets defined by the above constraints.
Furthermore, $U \cup V$ is a subset of the vertices of the axis-aligned box $[0,1]^{2+n}\times [-1,2]^m \subset \mathbb{R}^{2+n+m}$, and therefore $U\cup V$ is in convex position and form the vertices of $B$. To see that $B$ is $(n+m+1)$-dimensional, we simply note that \[\{u_{0,1},u_{0,2}, v_{0,1}, \dots, v_{0,n}, v_{1,1}, \dots, v_{1,m}\}\] is an affinely independent set of size $2+n+m$.

\subsubsection{The inner polytope $A$} \label{coord inner poly}
We now define vertices of the inner polytope $A$. These will consist of the points $U$, defined above, together with some additional points. 

For $1\leq i \leq m$ and $1\leq j \leq n$ define the points $y_{i,j}$ and $z_{i,j}$ that ensure Property \ref{seg-restr} are defined as 
\[y_{i,j} = \tfrac{1}{3}e_1 + \tfrac{2}{3}f_j + 2g_i -J =
(\tfrac{1}{3},0,\underbrace{0,\dots,\overset{j}{\tfrac{2}{3}},\dots,0}_n,\underbrace{-1,\dots,\overset{i}{1},\dots,-1}_m),\] and
\[z_{i,j} = \tfrac{1}{3}e_2 + \tfrac{2}{3}f_j + 2g_i -J =
(0,\tfrac{1}{3},\underbrace{0,\dots,\overset{j}{\tfrac{2}{3}},\dots,0}_n,\underbrace{-1,\dots,\overset{i}{1},\dots,-1}_m).\]

The point $p_{i,j,k}$ which encodes the equation $\alpha_{i,j}+\alpha_{i,k} = \frac{5}{2}$ is defined as
\[p_{i,j,k} = \tfrac{1}{2}f_j + \tfrac{1}{2}f_k+\tfrac{9}{4}g_i -J = (0,0,\underbrace{0,\dots,\overset{j}{\tfrac{1}{2}},\dots, \overset{k}{\tfrac{1}{2}}, \dots,0}_n,\underbrace{-1,\dots,\overset{i}{\tfrac{5}{4}},\dots,-1}_m).\]

The point $q_{i,j,k,l}$ which encodes the equation $\alpha_{i,j}+\alpha_{i,k} +\alpha_{i,l} = \frac{5}{2}$ is defined as 
\[q_{i,j,k,l} = \tfrac{1}{3}f_{j} + \tfrac{1}{3}f_{k} + \tfrac{1}{3}f_{l} + \tfrac{11}{6}g_{i}-J = 
(0,0,\underbrace{0,\dots,\overset{j}{\tfrac{1}{3}},\dots, \overset{k}{\tfrac{1}{3}}, \dots,\overset{l}{\tfrac{1}{3}}, \dots,0}_n,\underbrace{-1,\dots,\overset{i}{\tfrac{5}{6}},\dots,-1}_m).\]

The point $r_{i,j,k}$ which encodes the equation $\alpha_{i,k} \cdot \alpha_{j,k} = 1$ is defined as
\[r_{i,j,k} = f_k + g_i + g_j - J =
(0,0,\underbrace{0,\dots, \overset{k}{1}, \dots, 0}_n,\underbrace{-1,\dots,\overset{i}{0}, \dots,\overset{j}{0}, \dots ,\dots,-1}_m).\]

\subsection{Rational Equivalence}
Now, we describe the mapping $f:V(I) \rightarrow V(J)$.
Let $(\alpha_{i,j})_{(i,j)\in[m][n]}$ be a solution for the \ETRINVS/.
Then vertex $x_{i,j}$ is defined as 
\[v_{0,j} + (1 + \alpha_{i,j})[v_{i,j}- v_{0,j}].\]
All other vertices of the inner polytope are constant.
Thus $f$ is even a linear bijection.

\section{Conclusion}

One of the most compelling open questions is whether 
the extension complexity of a polytope can be computed 
in polynomial time.
It would be nice to get tight parametrized complexity bounds
for the \NPP/ problem.
The best parametrized algorithm for the \NMFshort/ runs in
$(nm)^{O(k^2)}$~\cite{MoitraAlmostOpt,moitra2016almost}. 
And by the exponential time hypothesis there is no 
$(nm)^{o(k)}$ algorithm~\cite{AroraNMFProvably}.
Another interesting direction, is the intermediate simplex
algorithm. We do know that it is NP-hard to compute
an intermediate simplex, but is it solvable in NP time?
At last, we want to point out that we also don't know 
\NP/-membership of \NPP/ for dimension $3$.
In a very recent line of research, two \ER/-hard problems where
shown to ``lie in \NP/'' under the ``lens of smoothed analysis''~\cite{SmoothedOrderType, smoothedArt}. 
It would be interesting to see, if a similar analysis
can be done with the nested polytope problem. 
It would be particular, interesting to see if it is
possible to develop algorithms using IP-solvers, as those
perform extremely well in practice.

\paragraph{Acknowledgement}
We would like to thank Anna Lubiw and Joseph O'Rourke, for 
helping us to find some relevant literature. 
We want to thank Mikkel Abrahamsen, for discussions.
Tillmann Miltzow acknowledges the generous support 
by the ERC Consolidator Grant 615640-ForEFront
and the Veni grant EAGER.

\bibliographystyle{plain}
\bibliography{references}

\begin{thebibliography}{10}

\bibitem{ARTETR}
Mikkel Abrahamsen, Anna Adamaszek, and Tillmann Miltzow.
\newblock The art gallery problem is {\ER/}-complete.
\newblock In {\em Symposium on Theory of Computing, {STOC} 2018}, pages 65--73,
  2018.
\newblock arxiv~1704.06969.

\bibitem{aggarwal1989finding}
Alok Aggarwal, Heather Booth, Joseph O'Rourke, Subhash Suri, and Chee~K. Yap.
\newblock Finding minimal convex nested polygons.
\newblock {\em Information and Computation}, 83(1):98--110, 1989.
\newblock also appeared at the first symposium on Computational geometry in
  1985.

\bibitem{AroraNMFProvably}
Sanjeev Arora, Rong Ge, Ravi Kannan, and Ankur Moitra.
\newblock Computing a nonnegative matrix factorization - provably.
\newblock {\em {SIAM} J. Comput.}, 45(4):1582--1611, 2016.
\newblock a preliminary version appeared at STOC 2012.

\bibitem{berman1973rank}
A.~Berman.
\newblock Rank factorization of nonnegative matrices.
\newblock {\em SIAM Review}, 15(3):655, 1973.

\bibitem{Bienstock91}
Daniel Bienstock.
\newblock Some provably hard crossing number problems.
\newblock {\em Discrete {\&} Computational Geometry}, 6:443--459, 1991.

\bibitem{SymmetricNash}
Vittorio Bil{\`o} and Marios Mavronicolas.
\newblock {Existential-R-Complete Decision Problems about Symmetric Nash
  Equilibria in Symmetric Multi-Player Games}.
\newblock In {\em 34th Symposium on Theoretical Aspects of Computer Science
  (STACS 2017)}, volume~66 of {\em Leibniz International Proceedings in
  Informatics (LIPIcs)}, pages 13:1--13:14, Dagstuhl, Germany, 2017. Schloss
  Dagstuhl--Leibniz-Zentrum fuer Informatik.

\bibitem{bronnimann1995almost}
Herv{\'e} Br{\"o}nnimann and Michael~T. Goodrich.
\newblock Almost optimal set covers in finite vc-dimension.
\newblock {\em Discrete \& Computational Geometry}, 14(4):463--479, 1995.

\bibitem{CannyPSPACE}
John Canny.
\newblock Some algebraic and geometric computations in pspace.
\newblock In {\em Proceedings of the Twentieth Annual ACM Symposium on Theory
  of Computing}, STOC '88, pages 460--467, New York, NY, USA, 1988. ACM.

\bibitem{cardinal2015computational}
Jean Cardinal.
\newblock Computational geometry column 62.
\newblock {\em ACM SIGACT News}, 46(4):69--78, 2015.

\bibitem{cardinal2017intersection}
Jean Cardinal, Stefan Felsner, Tillmann Miltzow, Casey Tompkins, and Birgit
  Vogtenhuber.
\newblock Intersection graphs of rays and grounded segments.
\newblock In {\em International Workshop on Graph-Theoretic Concepts in
  Computer Science}, pages 153--166. Springer, 2017.

\bibitem{cardinal2017recognition}
Jean Cardinal and Udo Hoffmann.
\newblock Recognition and complexity of point visibility graphs.
\newblock {\em Discrete \& Computational Geometry}, 57(1):164--178, 2017.

\bibitem{chistikov2017nonnegative}
Dmitry Chistikov, Stefan Kiefer, Ines Marusic, Mahsa Shirmohammadi, and James
  Worrell.
\newblock Nonnegative matrix factorization requires irrationality.
\newblock {\em SIAM Journal on Applied Algebra and Geometry}, 1(1):285--307,
  2017.
\newblock previous versions appeared at SODA~2017 and ICALP~2016.

\bibitem{clarkson1993algorithms}
Kenneth~L. Clarkson.
\newblock Algorithms for polytope covering and approximation.
\newblock In {\em Workshop on Algorithms and Data Structures}, pages 246--252.
  Springer, 1993.

\bibitem{cohen1993nonnegative}
Joel~E. Cohen and Uriel~G. Rothblum.
\newblock Nonnegative ranks, decompositions, and factorizations of nonnegative
  matrices.
\newblock {\em Linear Algebra and its Applications}, 190:149--168, 1993.

\bibitem{das1990approximation}
Gautam Das.
\newblock {\em Approximation schemes in computational geometry}.
\newblock PhD thesis, The University of Wisconsin-Madison, 1990.

\bibitem{das1997complexity}
Gautam Das and Michael~T. Goodrich.
\newblock On the complexity of optimization problems for 3-dimensional convex
  polyhedra and decision trees.
\newblock {\em Comput. Geom.}, 8(3):123--137, 1997.

\bibitem{dasCCCG}
Gautam Das and Deborah Joseph.
\newblock The complexity of minimum convex nested polyhedra.
\newblock In {\em Proc. 2nd Canad. Conf. Comput. Geom}, pages 296--301, 1990.

\bibitem{das1992minimum}
Gautam Das and Deborah Joseph.
\newblock Minimum vertex hulls for polyhedral domains.
\newblock {\em Theoretical computer science}, 103(1):107--135, 1992.

\bibitem{smoothedArt}
Michael~Gene Dobbins, Andreas Holmsen, and Tillmann Miltzow.
\newblock Smoothed analysis of the art gallery problem.
\newblock {\em arXiv:1811.01177}, 2018.

\bibitem{AreaKleist}
Michael~Gene Dobbins, Linda Kleist, Tillmann Miltzow, and Pawe{\l}
  Rz{a}{\.z}ewski.
\newblock $\forall \exists \mathbb{R}$-completeness and area-universality.
\newblock In {\em International Workshop on Graph-Theoretic Concepts in
  Computer Science}, pages 164--175. Springer, 2018.

\bibitem{SamNonSymmetric}
Samuel Fiorini, Serge Massar, Sebastian Pokutta, Hans~Raj Tiwary, and Ronald
  de~Wolf.
\newblock Linear vs. semidefinite extended formulations: exponential separation
  and strong lower bounds.
\newblock In {\em Proceedings of the 44th Symposium on Theory of Computing
  Conference, {STOC} 2012, New York, NY, USA, May 19 - 22, 2012}, pages
  95--106, 2012.

\bibitem{gillis2012geometric}
Nicolas Gillis and Fran{\c{c}}ois Glineur.
\newblock On the geometric interpretation of the nonnegative rank.
\newblock {\em Linear Algebra and its Applications}, 437(11):2685--2712, 2012.

\bibitem{KratochvilM94}
Jan Kratochv{\'{\i}}l and Ji{\v{r}}{\'{\i}} Matou{\v{s}}ek.
\newblock Intersection graphs of segments.
\newblock {\em Journal of Combinatorial Theory, Series {B}}, 62(2):289--315,
  1994.

\bibitem{lubiw2018complexity}
Anna Lubiw, Tillmann Miltzow, and Debajyoti Mondal.
\newblock The complexity of drawing a graph in a polygonal region.
\newblock In {\em International Symposium on Graph Drawing and Network
  Visualization}, pages 387--401. Springer, 2018.

\bibitem{matousekSegments}
Ji{\v{r}}{\'{\i}} Matou{\v{s}}ek.
\newblock Intersection graphs of segments and $\exists\mathbb{R}$.
\newblock {\em CoRR}, abs/1406.2636, 2014.

\bibitem{mcdiarmid2013integer}
Colin McDiarmid and Tobias M{\"u}ller.
\newblock Integer realizations of disk and segment graphs.
\newblock {\em Journal of Combinatorial Theory, Series B}, 103(1):114--143,
  2013.

\bibitem{Mnev}
Nikolai~E. Mn{\"e}v.
\newblock The universality theorems on the classification problem of
  configuration varieties and convex polytopes varieties.
\newblock In {\em Topology and geometry: Rohlin Seminar}, volume 1346 of {\em
  Lecture Notes in Mathematics}, pages 527--543, Berlin, 1988. Springer-Verlag.

\bibitem{MoitraAlmostOpt}
Ankur Moitra.
\newblock An almost optimal algorithm for computing nonnegative rank.
\newblock {\em {SIAM} J. Comput.}, 45(1):156--173, 2016.

\bibitem{moitra2016almost}
Ankur Moitra.
\newblock An almost optimal algorithm for computing nonnegative rank.
\newblock {\em SIAM Journal on Computing}, 45(1):156--173, 2016.
\newblock Appeared also at Soda 2013.

\bibitem{o1988computational}
Joseph O'Rourke.
\newblock The computational geometry column\# 4.
\newblock {\em ACM SIGGRAPH Computer Graphics}, 22(2):111--112, 1988.

\bibitem{richter1995realization}
J{\"u}rgen Richter-Gebert and G{\"u}nter~M Ziegler.
\newblock Realization spaces of 4-polytopes are universal.
\newblock {\em Bulletin of the American Mathematical Society}, 32(4):403--412,
  1995.

\bibitem{Schaefer09}
Marcus Schaefer.
\newblock Complexity of some geometric and topological problems.
\newblock In {\em Proceedings of the 17th International Symposium on Graph
  Drawing (GD)}, volume 5849 of {\em LNCS}, pages 334--344. Springer, 2010.

\bibitem{SchaeferS17}
Marcus Schaefer and Daniel \v{S}tefankovi\v{c}.
\newblock Fixed points, {N}ash equilibria, and the existential theory of the
  reals.
\newblock {\em Theory of Computing Systems}, 60(2):172--193, 2017.

\bibitem{shitov2016universality}
Yaroslav Shitov.
\newblock A universality theorem for nonnegative matrix factorizations.
\newblock {\em Preprint, \url{https://arxiv.org/abs/1606.09068}}, 2016.

\bibitem{shitov2017nonnegative}
Yaroslav Shitov.
\newblock The nonnegative rank of a matrix: Hard problems, easy solutions.
\newblock {\em SIAM Review}, 59(4):794--800, 2017.

\bibitem{shor1991stretchability}
Peter Shor.
\newblock Stretchability of pseudolines is np-hard.
\newblock {\em Applied Geometry and Discrete Mathematics-The Victor Klee
  Festschrift}, 1991.

\bibitem{silio1979efficient}
Charles~B. Silio~Jr.
\newblock An efficient simplex coverability algorithm in {$E^2$} with
  application to stochastic sequential machines.
\newblock {\em IEEE Trans. Computers}, 28(2):109--120, 1979.

\bibitem{Tarski}
Alfred Tarski.
\newblock A decision method for elementary algebra and geometry.
\newblock {\em Univ. of California Press}, 1951.
\newblock Berkeley.

\bibitem{SmoothedOrderType}
Ivor {van der Hoog}, Tillmann Miltzow, and Martijn van Schaik.
\newblock Smoothed analysis of order types.
\newblock {\em arXiv:1907.04645}, 2019.

\bibitem{vavasis2009complexity}
Stephen~A. Vavasis.
\newblock On the complexity of nonnegative matrix factorization.
\newblock {\em SIAM Journal on Optimization}, 20(3):1364--1377, 2009.

\bibitem{yannakakis1991expressing}
Mihalis Yannakakis.
\newblock Expressing combinatorial optimization problems by linear programs.
\newblock {\em Journal of Computer and System Sciences}, 43(3):441--466, 1991.

\end{thebibliography}

\newpage
\appendix
\section{Proof of Lemma~\ref{lem:CohenRothblum}}

\CohenReduction*

\begin{proof}
Let $M \in \R^{m\times n}_+$ be a non-negative matrix 
and $k\in \N$ be a number.
We define the outer polytope $B$ as the positive orthant in $\R^m$
intersected with the hyperplane 
\[H = \{x\in \R^{m} : \sum_i x_i = 1\}.\]
Note that $B$ can be specified by $m+1$  hyperplanes.
The inner polytope $A$ is defined as the convex hull
of all the columns of $M$. 
We have to show that there is a nested polytope
on $k$ vertices if and only if $\rank_+(M)\leq k$.
Let $X$ be a polytope with vertices $v_1,\ldots,v_k$
with $A\subseteq X\subseteq B$. 
For each column $c$ of $M$ exists 
$\lambda = (\lambda_1,\ldots,\lambda_k) \in [0,1]^k$
such that $\sum_i \lambda_i v_i = c$.
We define the matrix $V$ by the vectors $v_1,\ldots,v_n$.
For each column $c$ of $M$, we define the corresponding column
$\lambda$ of $W$. By definition $V\cdot W = M$ and both
$V$ and $W$ are non-negative.

For the reverse direction let $X$ be a polytope on $k$
vertices with $A\subseteq X\subseteq B$.
We define the columns of the matrix $V$ using 
the vertices of $X$, i.e., each vertex describes
exactly one column. Let $c$ be a column of $M$.
First note that we can assume $\|c\|_1 = 1$, i.e., $c\in H$.
Because, we can scale every column of $M$ and every column of
$W$ by some number $t>0$ without destroying or creating solutions.
Then it holds that $c\in X$ and more specifically
there is a vector $\lambda \in [0,1]^k$, with $\|\lambda\|_1 = 1$,
such that $V\cdot \lambda  = c$.
The column of $W$ corresponding to $c$ is $\lambda$.
This specifies the non-negative matrix factorization of
inner dimension $k$.
\end{proof}

\end{document}